\documentclass[review]{elsarticle}
\usepackage{}
\usepackage[latin1]{inputenc} 
\usepackage{listings, tcolorbox}
\usepackage{lineno,hyperref}
\usepackage{amssymb}
\usepackage{color}
\usepackage{amsmath}
\usepackage{tikz}
\usepackage{amsfonts}
\usepackage{mathrsfs}
\usepackage{amsthm}

\usepackage{bbding}
\usepackage{mathrsfs}
\usepackage{amsmath, amsfonts, amssymb, mathrsfs, txfonts}
\usepackage{graphicx, subfigure}
\usepackage{wasysym}
\usepackage{color}
\usepackage{bm}
\usepackage{enumerate}

\usepackage{pifont}
\usepackage{txfonts}
\usepackage{amsmath}
\usepackage{graphicx}
\usepackage{amsfonts}
\usepackage{amssymb}
\usepackage{mathrsfs,psfrag,eepic,epsfig}
\usepackage{epstopdf}

\biboptions{numbers,sort&compress}
\newtheorem{thm}{Theorem}[section]

\newtheorem{lem}[thm]{Lemma}

\theoremstyle{definition}

\modulolinenumbers[5]

\journal{Journal of \LaTeX\ Templates}

\makeatletter \@addtoreset{equation}{section}

\begin{document}

\begin{frontmatter}

\title{Long-time asymptotics for the Wadati-Konno-Ichikawa equation with the Schwartz initial data}
\tnotetext[mytitlenote]{Project supported by the Fundamental Research Fund for the Central Universities under the grant No. 2019ZDPY07.\\
\hspace*{3ex}$^{*}$Corresponding author.\\
\hspace*{3ex}\emph{E-mail addresses}: xwu@cumt.edu.cn (X. Wu), sftian@cumt.edu.cn,
shoufu2006@126.com (S. F. Tian)}

\author{Xin Wu and Shou-Fu Tian$^{*}$}
\address{
School of Mathematics, China University of Mining and Technology, Xuzhou 221116, People's Republic of China
}

\begin{abstract}
In this work, we investigate the long-time asymptotic behavior of the Wadati-Konno-Ichikawa equation with initial data belonging to Schwartz space at infinity by using the nonlinear steepest descent method of Deift and Zhou for the oscillatory Riemann-Hilbert problem. Based on the initial value condition, the original Riemann-Hilbert problem is constructed to express the solution of the Wadati-Konno-Ichikawa equation. Through a series of deformations, the original RH problem is transformed into a model RH problem, from which the long-time asymptotic solution of the equation is obtained explicitly.
\end{abstract}


\end{frontmatter}


\tableofcontents

\section{Introduction}

As a well-known integrable equation, the nonlinear Schr\"{o}dinger (NLS) equation
\begin{align}
iq_{t}+\frac{1}{2}q_{xx}+|q|^{2}q=0,
\end{align}
has played an important role in the field of mathematical physics \cite{1}-\cite{4}, which is widely used in optics, free water waves of ideal fluids, plasma waves, etc. When pulses with a short propagation time along the fiber are considered, the NLS equation with high-order nonlinear dispersion terms is further studied, which indeed plays an important role in great success of optical communication technology.
However, to high-order nonlinear dispersion, it not only does not expand the classes of Kerr-type nonlinearity, but also does not change the prediction of the NLS equation for an unlimited short pulse duration as the input field strength or the number of soliton increases.
Because the refractive index change caused by the field is controlled by high-order nonlinearity, from which we know that it cannot be described by Kerr-type nonlinearity for short pulses and higher input peak pulse power.
As a result, the photoinduced refractive index change becomes saturated at higher field strengths. In order to study the propagation of soliton in materials with saturable nonlinearity, the NLS equation with saturable nonlinearity \cite{Herrmann-1991,Porsezian-2012} is studied, which reads
\begin{align}
iq_{t}+q_{xx}+\frac{|q|^{2}}{1+r|q|^{2}}q=0,
\end{align}
where $r$ is a constant. However, it is difficult to obtain an analytical solution to this equation because it is not integrable.

For integrable models with saturable nonlinearities, people can trace back to the Wadati-Konno-Ichikawa (WKI) equation \cite{Wadati-1979,Shimizu-1980}, which is written as
\begin{align}\label{Q1}
iq_{t}+\left(\frac{q}{\sqrt{1+|q|^{2}}}\right)_{xx}=0,
\end{align}
where $q$ is a complex function of $x$ and $t$. As we all know, the WKI equation has very high nonlinearity and peak solutions. Therefore, more and more attention is paid to the research of this equation. For the WKI equation, there is a gauge transformation that enables mutual conversion between the Ablowitz, Kaup, Newell and Segur (AKNS) system and the WKI equation in \cite{Ishimori-1982}. Resorting to the gauge transformation,  the B\"{a}cklund transformation was studied in \cite{Boiti-1986}-\cite{Xiao-1990}. The classic inverse scattering method with vanishing boundary condition and non-vanishing boundary condition for the WKI equation was studied in \cite{Shimizu-1980} and \cite{Choudhury-1982}, respectively, and one soliton solution was obtained.
Exact stationary solution and the orbital stability for stationary solution of the WKI equation were studied in \cite{Van1-2013,Van2-2013}. The existence of global solution for the WKI equation with small initial data was studied in \cite{Shimabukuro-2016} and the algebro-geometric construction of WKI flows were studied in \cite{Li-2016}. In \cite{Zhang-2017}, the Darboux transformation was used to investigate the Wadati-Konno-Ichikawa system and the breathe and rogue wave solution were  obtained. Based on the Riemann-Hilbert problem which is the modern version, the WKI equation with simple poles and higher-order poles was considered in \cite{Liu-2017} and \cite{Zhang-2019}, and the phenomenon of abundant soliton solution has also been described.

In this work, we consider the WKI equation with the following initial value condition
\begin{align}\label{1.2}
q(x,0)=q_{0}(x)\in S(\mathbb{R}), \quad -\infty<x<+\infty,
\end{align}
where $S(\mathbb{R})$ means Schwartz space.

As early as 1976, the exact expression of the long time asymptotic solution for the NLS equation with decaying initial value was given by Zakharov and Manakov in \cite{Zakharov-1976}.  In which the inverse scattering map and the reconstruction of the solution was investigated through an
oscillatory Riemann-Hilbert problem. In 1993, Deift and Zhou proposed a nonlinear steepest descent method to analyze the asymptoticity of the solution of oscillatory Riemann-Hilbert problem, which gave a rigorous proof process \cite{Deift-1993}. This method mainly use a series of deformations to reduce the original Riemann-Hilbert problem into a Riemann-Hilbert problem whose jump matrix is asymptotic fast (as $t\rightarrow\infty$) to the identity matrix everywhere except near some stationary phase points. Due to its effectiveness in the long-time asymptotic behavior of the solution for nonlinear integral equations, the nonlinear steepest descent method is widely used in many equations,
such as the defocusing nonlinear Schr\"{o}dinger equation equation \cite{deNLS}, the sine-Gordon equation \cite{SG}, the modified nonlinear Schr\"{o}dinger equation \cite{mNLS1,mNLS2}, the focusing nonlinear Schr\"{o}dinger equation \cite{fNLS}, the Korteweg-de Vries equation \cite{KdV}, derivative nonlinear Schr\"{o}dinger equation equation \cite{dNLS}, the Cammasa-Holm equation \cite{CH}, Fokas-Lenells equation \cite{FL}, the Spin-1 Gross-Pitaevskii Equation \cite{GP}, coupled Hirota equation \cite{cHirota}, short pulse equation \cite{SP,CSP}, Kundu-Eckhaus Equation \cite{KE1,KE2}, the coupled dispersive AB system \cite{AB}, extended modified Korteweg-de Vries equation \cite{emKdV}, the ``good" Boussinesq equation \cite{gB}, nonlocal mKdV equation \cite{nmKdV} etc.

Our purpose is to study the asymptotic solution of the WKI equation by using the Deift-Zhou nonlinear steepest descent method when the time $t$ approaches infinity. Compared with the previous work, our work has some notable and valuable points. On the one hand, this work needs to consider $\lambda\rightarrow0$ and $\lambda\rightarrow\infty$ in the spectrum analysis compared with the Lax pair of the equation which is of the AKNS type. Therefore, if we want to use the Riemann-Hilbert problem to solve the WKI equation, we need to establish a connection about them. On the other hand, we use a simpler calculation process compared with the work done by Deift and Zhou in 1993. The classic nonlinear steepest descent method uses a series of deformations. However, in our work, the jump contours are divided into several parts, and the leading term is only provided near the stationary point, while the remaining parts provide the error terms.

The structure of this work is given as follows.
In section 2, in the spectral analysis of the WKI equation, due to the particularity of the Lax pair, we divide the spectral analysis into two parts, namely $\lambda\rightarrow0$ and $\lambda\rightarrow\infty$, to discuss the properties of the spectral function separately. In section 3, we first construct the RH problem for the case of $\lambda\rightarrow\infty$, and establish the relationship between the spectral functions in the two cases of $\lambda\rightarrow0$ and $\lambda\rightarrow\infty$, so that the solution of the WKI equation can be expressed by the solution of the Riemann-Hilbert problem. The purpose of section 4 is to use the nonlinear steepest descent method to deform the original Riemann-Hilbert problem into a model Riemann-Hilbert problem. In which, we make a series of transformations and introduce a parabolic cylinder function. In section 5, the long time asymptotic solution of the WKI equation is accurately given. Finally, some conclusions and discussions are presented.

\section{Spectral analysis}

In this section, starting with the Lax pair of \eqref{Q1}, we analyze the related properties of the Jost functions $\Psi(x,t;\lambda)$ and the scattering data $S(\lambda)$, which are prepared for the next section.

The Wadati-Konno-Ichikawa equation admits the Lax pair
\begin{align}\label{2.1}
\begin{split}
&\Psi_{x}=U\Psi,\\
&\Psi_{t}=V\Psi,
\end{split}
\end{align}
where
\begin{gather*}
U=-i\lambda\sigma_{3}+\lambda Q,\\
V=\left(\begin{array}{cc}
    -\frac{2i\lambda^{2}}{\Phi} & \frac{2q}{\Phi}\lambda^{2}+i\lambda\left(\frac{q}{\Phi}\right)_{x} \\
    -\frac{2q^{*}}{\Phi}\lambda^{2}+i\lambda\left(\frac{q^{*}}{\Phi}\right)_{x} & \frac{2i\lambda^{2}}{\Phi}  \\
 \end{array}\right),\\
\Phi=\sqrt{1+|q|^{2}}, \quad
\sigma_{3}=\left(\begin{array}{cc}
    1 & 0 \\
    0 & -1  \\
 \end{array}\right), \quad
Q=\left(\begin{array}{cc}
    0 & q \\
    -q^{*} & 0  \\
 \end{array}\right),
\end{gather*}
the function $\Psi$ is a $2\times2$ matrix and $\lambda$ is the spectral parameter. According to the compatibility conditions, that is $U_{t}-V_{x}+[U,V]=0$, the Lax pair is equivalent to the Wadati-Konno-Ichikawa equation.

Since the Lax pair \eqref{2.1} of the WKI equation is not of AKNS type, spectral analysis is different from most situations. By analogy to the short pulse equation, we consider the case of spectral problem as $\lambda\rightarrow0$ and $\lambda\rightarrow\infty$.

\subsection{For $\lambda\rightarrow0$}

According to Fokas method, introducing a transformation as
\begin{align}\label{2.2}
\Psi(x,t;\lambda)=\psi(x,t;\lambda)e^{-i(\lambda x+2\lambda^{2}t)\sigma_{3}},
\end{align}
such that as $x\rightarrow\pm\infty$
\begin{align}
\psi(x,t;\lambda)\sim\mathbb{I}.
\end{align}
Correspondingly, the Lax pair \eqref{2.1} is equivalent to
\begin{align}\label{2.3}
\left\{\begin{aligned}
&\psi_{x}+i\lambda[\sigma_{3},\psi]=\overline{U}\psi,\\
&\psi_{t}+2i\lambda^{2}[\sigma_{3},\psi]=\overline{V}\psi,
\end{aligned}\right.
\end{align}
with
\begin{align*}
&\overline{U}=\lambda Q, \\
&\overline{V}=\left(\begin{array}{cc}
    2i\lambda^{2}\left(1-\frac{1}{\Phi}\right) & \frac{2q}{\Phi}\lambda^{2}+i\lambda\left(\frac{q}{\Phi}\right)_{x} \\
    -\frac{2q^{*}}{\Phi}\lambda^{2}+i\lambda\left(\frac{q^{*}}{\Phi}\right)_{x} & -2i\lambda^{2}\left(1-\frac{1}{\Phi}\right)  \\
 \end{array}\right).
\end{align*}
Normally, we transform the space-part of \eqref{2.1} into the linear Volterra integral equations for $\psi_{\pm}(x,t;\lambda)$ as follows
\begin{align}\label{2.4}
\left\{\begin{aligned}
&\psi_{+}(x,t;\lambda)=\mathbb{I}-\int^{+\infty}_{x}e^{-i\lambda(x-y)\hat{\sigma_{3}}}\overline{U}(y,t)\psi_{+}(y,t;\lambda)dy,\\
&\psi_{-}(x,t;\lambda)=\mathbb{I}+\int^{x}_{-\infty}e^{-i\lambda(x-y)\hat{\sigma_{3}}}\overline{U}(y,t)\psi_{-}(y,t;\lambda)dy,
\end{aligned}\right.
\end{align}
where $e^{\hat{\sigma}}A=e^{\sigma}Ae^{-\sigma}$.

Denote $\psi_{\pm}=[\psi_{\pm,1}, \psi_{\pm,2}]$, it can directly know that
the exponential factor $e^{-i\lambda(x-y)}$ is only contained in the the integral equation for the first column of $\psi_{-}$ which decays when $\lambda$ is in the upper half plane, i.e. $\lambda\in\{\lambda|Im\lambda>0\}$, and the exponential factor $e^{-i\lambda(x-y)}$ is only contained in the integral equation for the second column of $\psi_{+}$ which decays when $\lambda$ is in the upper half plane, i.e. $\lambda\in\{\lambda|Im\lambda>0\}$. Thus $\psi_{-,1}$ and $\psi_{+,2}$ allow analytic extensions into $\{\lambda|Im\lambda>0\}$. Similar to the above method, it can be seen that $\psi_{-,2}$ and $\psi_{+,1}$ can be analytically continued to lower half plane, i.e. $\{\lambda|Im\lambda<0\}$.

The asymptotic properties for the Jost solutions $\psi_{\pm}$ are derived as
\begin{align}
\psi_{\pm}(x,t;\lambda)=I+\int^{x}_{\pm\infty}\overline{U}(y,t;\lambda)dy+O(\lambda^{2}), \quad
as \quad \lambda\rightarrow0.
\end{align}

To begin with the asymptotic property of the Jost solutions $\psi_{\pm}$, the solution of WKI equation is presented as
\begin{align}
Q(x,t)=\lim_{\lambda\rightarrow0}\frac{1}{\lambda}\frac{\partial}{\partial x}\psi_{\pm}(x,t;\lambda).
\end{align}

\subsection{For $\lambda\rightarrow\infty$}

It is worth noting that for the two cases of $\lambda\rightarrow0$ and $\lambda\rightarrow\infty$, the asymptotic properties of the Jost solution are not universal when constructing the RH problem. Therefore, when $\lambda\rightarrow\infty$, we take a different approach than when constructing the RH problem described above. In what follows, the transformation is need to be introduced to control the eigenfunctions.

To begin with, we introduce P(x,t) which is a $2\times2$ matrix-value function as
\begin{align}
P(x,t)=\sqrt{\frac{1+\Phi}{2\Phi}}\left(\begin{array}{cc}
                                    1 & \frac{-i(1-\Phi)}{q^{*}} \\
                                    \frac{-i(1-\Phi)}{q} & 1
                                  \end{array}\right),
\end{align}
then its inverse matrix is
\begin{align}
P^{-1}(x,t)=\sqrt{\frac{1+\Phi}{2\Phi}}\left(\begin{array}{cc}
                                    1 & \frac{i(1-\Phi)}{q^{*}} \\
                                    \frac{i(1-\Phi)}{q} & 1
                                  \end{array}\right).
\end{align}

Defining a transformation as
\begin{align}
\tilde{\Psi}(x,t;\lambda)=P(x,t)\Psi(x,t;\lambda),
\end{align}
therefore the Lax pair \eqref{2.1} is transformed into
\begin{align}\label{2.11}
\left\{\begin{aligned}
&\tilde{\Psi}_{x}+\theta_{x}\tilde{\Psi}=\tilde{U}\tilde{\Psi},\\
&\tilde{\Psi}_{t}+\theta_{t}\tilde{\Psi}=\tilde{V}\tilde{\Psi},
\end{aligned}\right.
\end{align}
here
\begin{align}\label{2.12}
&\theta_{x}=i\lambda\Phi\sigma_{3},\notag\\
&\theta_{t}=2i\lambda^{2}\sigma_{3}+\lambda\frac{qq^{*}_{x}-q_{x}q^{*}}{2\Phi^{2}}\sigma_{3},
\end{align}
with
\begin{gather*}
\tilde{U}=-P^{-1}P_{x}
=\left(\begin{array}{cc}
 -\frac{qq^{*}_{x}-q_{x}q^{*}}{4\Phi(1+\Phi)} & -\frac{iq[\Phi(qq^{*}_{x}-q_{x}q^{*})-|q|^{2}_{x}]}{4\Phi^{2}(\Phi^{2}-1)} \\
 \frac{iq^{*}[\Phi(qq^{*}_{x}-q_{x}q^{*})+|q|^{2}_{x}]}{4\Phi^{2}(\Phi^{2}-1)} & -\frac{qq^{*}_{x}-q_{x}q^{*}}{4\Phi(1+\Phi)}
                                  \end{array}\right),\\
\tilde{V}=\left(\begin{array}{cc}
                                    \tilde{V}_{11} & \tilde{V}_{12} \\
                                    \tilde{V}_{21} & -\tilde{V}_{11}
                                  \end{array}\right),\\
\tilde{V}_{11}=-\frac{qq^{*}_{t}-q_{t}q^{*}}{4\Phi(1+\Phi)},\\
\tilde{V}_{12}=-\frac{iq[q^{*}(qq^{*}_{x}-q^{*}q_{x})-2|q|^{*}_{x}(\Phi-1)]}{2(\Phi-1)\Phi^{3}q^{*}}\lambda
-\frac{iq[q^{*}(qq^{*}_{t}-q_{t}q^{*})-2|q|^{*}_{t}(\Phi-1)]}{4\Phi^{2}(\Phi^{2}-1)q^{*}},\\
\tilde{V}_{21}=\frac{iq^{*}[q(qq^{*}_{x}-q^{*}q_{x})+2|q|_{x}(\Phi-1)]}{2(\Phi-1)\Phi^{3}q}\lambda
+\frac{iq^{*}[q(qq^{*}_{t}-q_{t}q^{*})+2|q|_{t}(\Phi-1)]}{4\Phi^{2}(\Phi^{2}-1)q}.
\end{gather*}
It is easy to see that \eqref{2.12} is compatible because of $\theta_{xt}=\theta_{tx}$ leading to
\begin{align}
i\Phi_{t}=\left(\frac{qq^{*}_{x}-q^{*}q_{x}}{2\Phi^{2}}\right)_{x},
\end{align}
which is the conservation law of the WKI equation. For this reason, the function $\theta(x,t;\lambda)$ is defined as
\begin{align}
\theta(x,t;\lambda)=(i\lambda\tilde{x}+2i\lambda^{2}t)\sigma_{3},
\end{align}
with
\begin{align}
\tilde{x}(x,t)=x-\int^{+\infty}_{x}(\Phi(y)-1)dy.
\end{align}

Defining
\begin{align}
\phi(x,t;\lambda)=e^{d_{+}\hat{\sigma}_{3}}\tilde{\Psi}e^{-d_{-}\sigma_{3}}e^{\theta},
\end{align}
where
\begin{gather*}
d_{+}=\int^{x}_{-\infty}\frac{qq^{*}_{x}-q^{*}q_{x}}{4\Phi(1+\Phi)}(\xi,t)d\xi,\\
d_{-}=\int^{+\infty}_{x}\frac{qq^{*}_{x}-q^{*}q_{x}}{4\Phi(1+\Phi)}(\xi,t)d\xi,\\
d=d_{+}+d_{-}=\int^{+\infty}_{-\infty}\frac{qq^{*}_{x}-q^{*}q_{x}}{4\Phi(1+\Phi)}(\xi,t)d\xi.
\end{gather*}
Then
\begin{align}
\left\{\begin{aligned}
&\phi_{x}+[\theta_{x},\phi]=\tilde{U}\phi,\\
&\phi_{t}+[\theta_{t},\phi]=\tilde{V}\phi,
\end{aligned}\right.
\end{align}
By using the same way, we have
\begin{align}\label{2.18}
\phi_{\pm}(x,t;\lambda)=I+\int^{x}_{\pm\infty}e^{\theta(y,t;\lambda)-\theta(x,t;\lambda)}
\tilde{U}(y,t;\lambda)\phi_{\pm}(y,t;\lambda)e^{\theta(x,t;\lambda)-\theta(y,t;\lambda)}dy.
\end{align}
If the definition of $\theta$ is considered, the above Volterra integral equations \eqref{2.18} are written as
\begin{align}\label{2.19}
\phi_{\pm}(x,t;\lambda)=I+\int^{x}_{\pm\infty}e^{-i\lambda\int^{x}_{y}\Phi(\xi)d\xi\sigma_{3}}
\tilde{U}(y,t;\lambda)\phi_{\pm}(y,t;\lambda)e^{i\lambda\int^{x}_{y}\Phi(\xi)d\xi\sigma_{3}}dy.
\end{align}

On the basis of \eqref{2.19} and the above analysis, the Jost functions $\phi_{\pm}(x,t;\lambda)$ hold the following analytic properties as:
$\phi_{+,1}$ and $\phi_{-,2}$ (and, respectively, $\phi_{-,1}$ and $\phi_{+,2}$) are analytic in $\{\lambda|Im\lambda<0\}$ (and, respectively, $\{\lambda|Im\lambda>0\}$) and continuous in $\{\lambda|Im\lambda\leq0\}$ (and, respectively, $\{\lambda|Im\lambda\geq0\}$).

Meanwhile, the Jost functions $\phi_{\pm}(x,t;\lambda)$ satisfy the symmetry relationship as
\begin{align}
\phi^{*}_{\pm}(x,t;\lambda^{*})=\left(\begin{array}{cc}
                                       0 & 1 \\
                                       -1 & 0
                                     \end{array}\right)
\phi_{\pm}(x,t;\lambda)\left(\begin{array}{cc}
                                       0 & -1 \\
                                       1 & 0
                                     \end{array}\right),
\end{align}
and the asymptotic properties as:\\
\begin{itemize}
\item $\left(\phi_{+,1}, \phi_{-,2}\right)\rightarrow I$ as $\lambda\rightarrow\infty$ in $\{\lambda|Im\lambda\leq0\}$;\\
\item $\left(\phi_{-,1}, \phi_{+,2}\right)\rightarrow I$ as $\lambda\rightarrow\infty$ in $\{\lambda|Im\lambda\geq0\}$.
\end{itemize}
Since $\tilde{\Psi}_{\pm}$ are the solution of \eqref{2.11}, there is a linear relationship between them, then there is a matrix $S(\lambda)$ which is not depend on $x$ and $t$ such that
\begin{align}\label{2.21}
\phi_{+}(x,t;\lambda)=\phi_{-}(x,t;\lambda)e^{-(i\lambda\tilde{x}+2i\lambda^{2}t)\sigma_{3}}S(\lambda)
e^{(i\lambda\tilde{x}+2i\lambda^{2}t)\sigma_{3}},
\end{align}
where $S(\lambda)=(s_{ij})_{2\times2}$. According to the symmetry property of $\phi_{\pm}$, $S(\lambda)$ can be rewritten as
\begin{align*}
S(\lambda)=\left(\begin{array}{cc}
                         a^{*}(\lambda) & b(\lambda) \\
                         -b^{*}(\lambda) & a(\lambda)
                       \end{array}\right), \quad \lambda\in\mathbb{R},
\end{align*}
where
\begin{align}
&a(\lambda)=\det(\phi_{-,1},\phi_{+,2}),\notag\\
&b(\lambda)=e^{2i(\lambda\tilde{x}+2i\lambda^{2}t)}\det(\phi_{+,2},\phi_{-,2}).
\end{align}

Furthermore, for all $\lambda\in\mathbb{R}$, one has
\begin{align}
\det S(\lambda)=1.
\end{align}

Based on the properties of $\phi_{\pm}$, we know that scattering data, that is $a(\lambda)$ and $b(\lambda)$, admit the asymptotic properties as:
\begin{itemize}
\item $a(\lambda)\rightarrow1$ as $\lambda\rightarrow\infty$;
\item $b(\lambda)\rightarrow0$ as $\lambda\rightarrow\infty$;
\end{itemize}
and the analytic properties as:
\begin{itemize}
\item $a(\lambda)$ is analytic in $\{\lambda|Im\lambda<0\}$ and continuous in $\{\lambda|Im\lambda\leq0\}$;
\item $b(\lambda)$ is are nowhere analytic and continuous for $\lambda\in\mathbb{R}$.
\end{itemize}
\section{The Riemann-Hilbert problem}

In accordance with the properties (the analytic properties and the asymptotic properties) of the Jost functions $\phi_{\pm}(x,t;\lambda)$ and scattering data $S(\lambda)$, we define a piecewise meromorphic $2\times2$ matrix function as
\begin{align}
M(x,t;\lambda)=\left\{\begin{aligned}
&M_{+}(x,t;\lambda)=\left(\phi_{-,1}(x,t;\lambda), \frac{\phi_{+,2}(x,t;\lambda)}{a(\lambda)}\right),\quad Im\lambda>0,\\
&M_{-}(x,t;\lambda)=\left(\frac{\psi_{+,1}(x,t;\lambda)}{a^{*}(\lambda^{*})}, \psi_{-,2}(x,t;\lambda)\right),\quad Im\lambda<0,
\end{aligned}\right.
\end{align}
with $M_{\pm}(x,t;\lambda)=\mathop{\lim}\limits_{\varepsilon\rightarrow0^{+}}M(x,t;\lambda\pm i\varepsilon)$ for $\varepsilon, \lambda\in\mathbb{R}$, as well as the scattering coefficient as
\begin{align}
r(k)=\frac{b(\lambda)}{a(\lambda)}.
\end{align}

From the preparation of the previous section, we have the following theorem.
\begin{thm}
The piecewise meromorphic function $M(x,t;\lambda)$ satisfies the following Riemann-Hilbert problem as:
\begin{itemize}
\item $M(x,t;\lambda)$ is analytic in $\mathbb{C}\backslash\mathbb{R}$.
\item The jump condition is
\begin{align}
M_{+}(x,t;\lambda)=M_{-}(x,t;\lambda)e^{-(i\lambda\tilde{x}+2i\lambda^{2}t)\hat{\sigma}_{3}}J(\lambda),
\end{align}
with
\begin{align*}
J(\lambda)=\left(\begin{array}{cc}
                 1 & r(\lambda) \\
                 r^{*}(\lambda) & 1+|r(\lambda)|^{2}
               \end{array}\right).
\end{align*}
\item $M(x,t;\lambda)\rightarrow I$, \quad as $\lambda\rightarrow\infty$.
\end{itemize}
\end{thm}

Considering the Jost functions $\phi_{\pm}(x,t;\lambda)$, we have
\begin{align}
\det M(x,t;\lambda)=1.
\end{align}
Simultaneously, the symmetry property exists the relationship as
\begin{align}
M^{*}(x,t;\lambda^{*})=\left(\begin{array}{cc}
                                       0 & 1 \\
                                       -1 & 0
                                     \end{array}\right)
M(x,t;\lambda)\left(\begin{array}{cc}
                                       0 & -1 \\
                                       1 & 0
                                     \end{array}\right).
\end{align}

For the jump matrix in the jump condition, it is easy to see that it is not only related to the scattering data, i.e. $a(\lambda)$ and $b(\lambda)$, which is obtained because of the initial data $q(x,0)$, but also related to the unknown $\tilde{x}$. Therefore, we need to convert the plane under consideration from $(x,t)$ to $(\tilde{x},t)$.

With regard to $\tilde{M}(\tilde{x},t;\lambda)=M(\tilde{x}(x,t),t;\lambda)$, one has
\begin{align}
\tilde{M}_{+}(\tilde{x},t;\lambda)=\tilde{M}_{-}(\tilde{x},t;\lambda)\tilde{J}(\tilde{x},t;\lambda),
\quad \lambda\in\mathbb{R},
\end{align}
with $\tilde{J}(\tilde{x},t;\lambda)=e^{-(i\lambda\tilde{x}+2i\lambda^{2}t)\hat{\sigma}_{3}}J(\lambda)$.

Besides, the function $\tilde{M}(\tilde{x},t;\lambda)$ admits the normalization condition
\begin{align}
\tilde{M}(\tilde{x},t;\lambda)\rightarrow I,\quad as~ \lambda\rightarrow\infty,
\end{align}
which guarantees the uniqueness of the Riemann-Hilbert solution.

Because we use $\phi_{\pm}(x,t;\lambda)$ to introduce $M(x,t;\lambda)$, which is the solution of RH problem, however when constructing the solution of the equation, it is provided by $\psi_{\pm}(x,t;\lambda)$, so a connection between $\phi_{\pm}(x,t;\lambda)$ and $\psi_{\pm}(x,t;\lambda)$ is established.

It is noted that both $\phi_{\pm}(x,t;\lambda)$ and $\psi_{\pm}(x,t;\lambda)$ are solution of the Lax pair \eqref{2.1}, there exists the linear relationship as
\begin{align}
\phi_{\pm}(x,t;\lambda)=e^{-d_{+}\sigma_{3}}P^{-1}(x,t)\psi_{\pm}(x,t;\lambda)e^{-(i\lambda x+2i\lambda^{2}t)\sigma_{3}}C_{\pm}(\lambda)e^{\theta}e^{d\sigma_{3}},
\end{align}
where $\theta=(i\lambda\tilde{x}+2i\lambda^{2}t)\sigma_{3}$ and $C_{\pm}(\lambda)$ are only depend on $\lambda$. Through direct calculation, it is known that when $x\rightarrow\pm$, there exists
\begin{align}
C_{+}(\lambda)=e^{-d\sigma_{3}}e^{i\lambda c\sigma_{3}},\quad C_{-}(\lambda)=I,
\end{align}
where $c=\int^{+\infty}_{-\infty}(\Phi(\xi)-1)d\xi$ is a quantity conserved under the dynamics governed by WKI equation, namely, $c$ is a constant.

Then
\begin{align}
&\phi_{+}(x,t;\lambda)=e^{-d_{+}\sigma_{3}}P^{-1}(x,t)\psi_{+}(x,t;\lambda)
e^{i\lambda\int^{x}_{+\infty}(\Phi(y)-1)dy\sigma_{3}},\\
&\phi_{-}(x,t;\lambda)=e^{-d_{+}\sigma_{3}}P^{-1}(x,t)\psi_{-}(x,t;\lambda)
e^{i\lambda\int^{x}_{-\infty}(\Phi(y)-1)dy\sigma_{3}}e^{d\sigma_{3}}.
\end{align}

In the light of the asymptotic properties to the Jost functions $\psi_{\pm}(x,t;\lambda)$ and the scattering data, when $\lambda\rightarrow0$ we have
\begin{align*}
M_{+}(x,t;\lambda)=&e^{-d_{+}\sigma_{3}}P^{-1}(x,t)\left[
I+\int^{x}_{+\infty}\overline{U}(y,t;\lambda)dy+i\lambda\int^{x}_{+\infty}(\Phi(y)-1)dy\sigma_{3}\right]\\
&\times e^{i\lambda\int^{x}_{+\infty}(\Phi(y)-1)dy\sigma_{3}}+O(\lambda^{2}),\\
M_{-}(x,t;\lambda)=&e^{-d_{+}\sigma_{3}}P^{-1}(x,t)\left[
I+\int^{x}_{-\infty}\overline{U}(y,t;\lambda)dy+i\lambda\int^{x}_{-\infty}(\Phi(y)-1)dy\sigma_{3}\right]\\
&\times e^{i\lambda\int^{x}_{-\infty}(\Phi(y)-1)dy\sigma_{3}}e^{d\sigma_{3}}+O(\lambda^{2}),
\end{align*}
and
\begin{align*}
&a(\lambda)=e^{d}\left(1+i\lambda c+O(\lambda^{2})\right),\\
&b(\lambda)=O(\lambda).
\end{align*}

In the above, we know that the RH problem about $M(x,t;\lambda)$ is transformed into the RH problem about $\tilde{M}(\tilde{x},t;\lambda)$, and the solution $q(x,t)$ of the WKI equation with the initial value problem can be given by $\tilde{M}(\tilde{x},t;\lambda)$, that is
\begin{align}
q(x,t)=\tilde{q}(\tilde{x}(x,t),t),
\end{align}
where
\begin{align}
&\tilde{x}(x,t)=\tilde{x}+\lim_{\lambda\rightarrow0}\frac{\left(\tilde{M}(\tilde{x},t;0)\right)^{-1}
\left(\tilde{M}(\tilde{x},t;\lambda)\right)_{11}-1}{i\lambda},\\
&e^{-2d}q(\tilde{x},t)=\lim_{\lambda\rightarrow0}\frac{\partial}{\partial\tilde{x}}
\frac{\left(\left(\tilde{M}(\tilde{x},t;0)\right)^{-1}\tilde{M}(\tilde{x},t;\lambda)\right)_{12}}{i\lambda}.
\end{align}

\section{Reduction to a model Riemann-Hilbert problem}

In this section, we convert the original RH problem into the solved RH problem in which the jump matrix is asymptotic to an identity matrix when $t\to\infty$ except near the stationary points and introduce a solvable parabolic cylinder model, so as to obtain a model Riemann-Hilbert problem which long time asymptotic solution is solved.

Note that the jump matrix $\tilde{J}(\tilde{x},t;\lambda)$ has two oscillation terms $e^{-(i\lambda\tilde{x}+2i\lambda^{2}t)\sigma_{3}}$ and $e^{(i\lambda\tilde{x}+2i\lambda^{2}t)\sigma_{3}}$. Let
\begin{align}
\varphi(\lambda)=i\lambda\frac{\tilde{x}}{t}+2i\lambda^{2}
\triangleq i\lambda\zeta+2i\lambda^{2},
\end{align}
we discuss the stationary point of $\varphi(\lambda)$, that is, taking
\begin{align}
\frac{d\varphi(\lambda)}{d\lambda}=0,
\end{align}
thus the stationary point is solved as $\lambda_{0}=-\frac{4\tilde{x}}{t}$. Then $\varphi(\lambda)$ is rewritten as
\begin{align}\label{5.3}
\varphi(\lambda)=2i\lambda^{2}-4i\lambda\lambda_{0}=2i(\lambda-\lambda_{0})^{2}-2i\lambda_{0}^{2}.
\end{align}

On the basis of the above result and the nonlinear steepest descent method which was invented by Deift and Zhou, we begin to deform the original Riemann-Hilbert problem as follows.

{\bf Step 1}
The purpose of the first step is to find the triangular factorization of the jump matrix. If not, it is achieved by introducing the scalar Riemann-Hilbert problem.

Taking $\lambda-\lambda_{0}=\rho e^{i\beta}$ and substituting it into \eqref{5.3}, one has
\begin{align}
\varphi(\lambda)=i\left(2\rho^{2}\cos2\beta-2\lambda^{2}_{0}\right)-2\rho^{2}\sin2\beta,
\end{align}
then
\begin{align}
Re\left(\varphi(\lambda)\right)=-2\rho^{2}\sin2\beta.
\end{align}
It is easy to see that $\sin2\beta>0$ when $0<\beta<\pi/2$ or $\pi<\beta<3\pi/2$, thus $Re\left(\varphi(\lambda)\right)<0$. On the other hand, when $\pi/2<\beta<\pi$ or $3\pi/2<\beta<2\pi$,  we have $\sin2\beta<0$, so that $Re\left(\varphi(\lambda)\right)>0$. That is shown as Fig.1.

\centerline{\begin{tikzpicture}[scale=0.5]
\path [fill=pink] (0,0)--(5,0) to (5,5) -- (0,5);
\path [fill=pink] (0,0)--(-5,0) to (-5,-5) -- (0,-5);
\draw[->][thick](-5,0)--(5,0)[thick]node[right]{$Rez$};
\draw[->][thick](0,-5)--(0,5)[thick]node[above]{$Imz$};
\draw[fill] (2.5,2.5)node[below]{$D_{1}$};
\draw[fill] (2.5,-2.5)node[below]{$D_{2}$};
\draw[fill] (-2.5,2.5)node[below]{$D_{4}$};
\draw[fill] (-2.5,-2.5)node[below]{$D_{3}$};
\end{tikzpicture}}
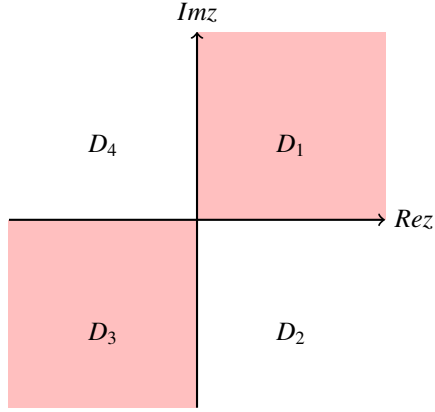
\centerline{\noindent {\small \textbf{Fig. 1} Exponential decaying domains.}}

According to the analysis of $\varphi(\lambda)$, the jump matrix $\tilde{J}(\tilde{x},t;\lambda)$ enjoys the lower/upper triangular factorization as
\begin{align}
\tilde{J}(\tilde{x},t;\lambda)=
\left(\begin{array}{cc}
      1 & 0 \\
      r^{*}(\lambda)e^{2t\varphi(\lambda)} & 1 \\
        \end{array}\right)\left(\begin{array}{cc}
      1 & r(\lambda)e^{-2t\varphi(\lambda)} \\
      0 & 1 \\
        \end{array}\right),
\end{align}
and the upper/diagonal/lower factorization as
\begin{align}
\tilde{J}(\tilde{x},t;\lambda)=
&\left(\begin{array}{cc}
      1 & \frac{r(\lambda)}{1+|r(\lambda)|^{2}}e^{-2t\varphi(\lambda)} \\
      0 & 1 \\
        \end{array}\right)\left(\begin{array}{cc}
      \frac{1}{1+|r(\lambda)|^{2}} & 0 \\
      0 & 1+|r(\lambda)|^{2} \\
        \end{array}\right)\notag\\
        &\times \left(\begin{array}{cc}
      1 & 0 \\
      \frac{r^{*}(\lambda)}{1+|r(\lambda)|^{2}}e^{2t\varphi(\lambda)} & 1 \\
        \end{array}\right).
\end{align}

It is noticed that there is a diagonal matrix in the second decomposition, so we introduce a scalar RH problem as
\begin{itemize}
\item $\delta(\lambda)$ is analytic in $\mathbb{C}\setminus\mathbb{R}$;
\item the jump condition is
\begin{align*}
\delta_{+}(\lambda)=\left\{\begin{aligned}
&\delta_{-}(\lambda)(1+|r(\lambda)|^{2}), \quad \lambda<\lambda_{0},\\
&\delta_{+}(\lambda)=\delta_{-}, \quad \lambda>\lambda_{0};
\end{aligned}\right.
\end{align*}
\item $\delta(\lambda)\rightarrow1$ as $\lambda\rightarrow\infty$.
\end{itemize}

According to the Plemelj's formula, $\delta(\lambda)$ can be solved as
\begin{align}
\delta(\lambda)=\exp\left[\frac{1}{2\pi i}\int^{\lambda_{0}}_{-\infty}\frac{\log(1+|r(s)|^{2})}{s-\lambda}ds\right].
\end{align}

Furthermore, $\delta(\lambda)$ can be rewritten as
\begin{align}
\delta(\lambda)=(\lambda-\lambda_{0})^{i\upsilon}e^{\chi(\lambda)},
\end{align}
where
\begin{align*}
&\upsilon=-\frac{1}{2\pi}\log(1+|r(\lambda_{0})|^{2}),\\
&\chi\left(\lambda\right)=-\frac{1}{2\pi i}\int^{\lambda_{0}}_{-\infty}\log(\lambda-s)d\log(1+|r(s)|^{2}).
\end{align*}

Based on the above analysis, define the transformation as
\begin{align}
M^{(1)}(\tilde{x},t;\lambda)=\tilde{M}(\tilde{x},t;\lambda)\delta(\lambda)^{\sigma_{3}},
\end{align}
then $M^{(1)}(\tilde{x},t;\lambda)$ admits the RH problem as follows
\begin{itemize}
\item $M^{(1)}(\tilde{x},t;\lambda)$ is analytic in $\mathbb{C}\backslash\mathbb{R}$.
\item The jump condition is
\begin{align}
M^{(1)}_{+}(\tilde{x},t;\lambda)=M^{(1)}_{-}(\tilde{x},t;\lambda)J^{(1)}(\tilde{x},t;\lambda),
\end{align}
with
{\small \begin{align*}
J^{(1)}(\tilde{x},t;\lambda)=\left\{\begin{aligned}
&\left(\begin{array}{cc}
      1 & 0 \\
      e^{2t\varphi(\lambda)}r^{*}(\lambda)\delta^{2}(\lambda) & 1 \\
        \end{array}\right)\left(\begin{array}{cc}
      1 & e^{-2t\varphi(\lambda)}r(\lambda)\delta^{-2}(\lambda) \\
      0 & 1 \\
        \end{array}\right), \quad \lambda>\lambda_{0},\\
&\left(\begin{array}{cc}
      1 & e^{-2t\varphi(\lambda)}\frac{r(\lambda)}{1+|r(\lambda)|^{2}}\delta^{-2}_{-}(\lambda) \\
      0 & 1 \\
        \end{array}\right)\left(\begin{array}{cc}
      1 & 0 \\
      e^{2t\varphi(\lambda)}\frac{r^{*}(\lambda)}{1+|r(\lambda)|^{2}}\delta^{2}_{+}(\lambda) & 1 \\
        \end{array}\right), \quad \lambda<\lambda_{0}.
               \end{aligned}\right.
\end{align*}}
\item $M^{(1)}_{+}(\tilde{x},t;\lambda)\rightarrow I$, \quad as $\lambda\rightarrow\infty$.
\end{itemize}

{\bf Step 2}
The second step is to deform the contour, analytic approximations about $r(\lambda)$ are firstly introduced.

For simplicity, denote
\begin{align*}
& r(\lambda)\doteq r_{1}(\lambda),\quad r^{*}(\lambda)\doteq r_{2}(\lambda),\notag\\
&
\frac{r(\lambda)}{1+|r(\lambda)|^{2}}\doteq r_{3}(\lambda),\quad
\frac{r^{*}(\lambda)}{1+|r(\lambda)|^{2}}\doteq r_{4}(\lambda).
\end{align*}

As shown in Fig. 1, the complex $\lambda$-plane is divided into four parts, i.e. $D_{j}, j=1,2,3,4,$ so that
\begin{align*}
\{\lambda|Re\left(\varphi(\lambda)\right)<0\}=D_{1}\cup D_{3},\quad
\{\lambda|Re\left(\varphi(\lambda)\right)>0\}=D_{2}\cup D_{4}.
\end{align*}
Making the decomposition as follows
\begin{align}
r_{j}(\lambda)=\left\{\begin{aligned}
&r_{j,a}(\tilde{x},t;\lambda)+r_{j,r}(\tilde{x},t;\lambda),\quad j=1,2,|\lambda|>\lambda_{0},\quad \lambda\in\mathbb{R},\\
&r_{j,a}(\tilde{x},t;\lambda)+r_{j,r}(\tilde{x},t;\lambda),\quad j=3,4,|\lambda|<\lambda_{0},\quad \lambda\in\mathbb{R}.
\end{aligned}\right.
\end{align}

Due to $r(\lambda)\in C^{11}(\mathbb{R}_{+})$, expanding on it obtains
\begin{align}
&r^{(n)}_{1}(\lambda)=\frac{d^{n}}{d\lambda^{n}}\left(\sum^{6}_{j=0}p_{j}(\zeta)(\lambda-\lambda_{0})^{j}\right)
+O\left((\lambda-\lambda_{0})^{7-n}\right), \notag\\
 &\lambda\rightarrow\lambda_{0}, \quad n=0,1,2,
\end{align}
with $p_{j}(\zeta)=r^{(j)}_{1}(\lambda_{0})/j!$.

Let
\begin{align}
g_{0}(\lambda)=\sum^{10}_{j=4}\frac{a_{j}}{(\lambda-i)^{j}},
\end{align}
where $\{a_{j}\}^{10}_{4}$ make
\begin{align}\label{5.15}
g_{0}(\lambda)=\sum^{6}_{j=0}p_{j}(\zeta)(\lambda-\lambda_{0})^{j}+O\left((\lambda-\lambda_{0})^{7-n}\right),
 \quad \lambda\rightarrow\lambda_{0}.
\end{align}
Equation \eqref{5.15} imposes seven linearly independent conditions on the $a_{j}$, which can be verified directly. Thus the coefficients $a_{j}$ not only exist but are unique.

Making $g(\lambda)=r_{1}(\lambda)-g_{0}(\lambda)$, then

(i) $g_{0}(\lambda)$ is a rational function of $\lambda\in\mathbb{C}$, which has no poles in $\bar{D}_{1}$;

(ii) $g_{0}(\lambda)$ is consistent with $r_{1}(\lambda)$ to six order at $\lambda_{0}$ and to three order at $\infty$, that is
\begin{align}
\frac{d^{n}}{d\lambda^{n}}g(\lambda)=\left\{\begin{aligned}
&O\left((\lambda-\lambda_{0})^{7-n}\right),\quad \lambda\rightarrow\lambda_{0},\\
&O\left(\lambda^{-4+2n}\right),\quad \lambda\rightarrow\infty,
\end{aligned}\right.
\end{align}
with $\lambda\in\mathbb{R}$ and $n=0,1,2$.

Next, we begin to analyze the decomposition of $r_{1}(\lambda)$. The map $\lambda\rightarrow\tau$ is a bijection $[\lambda_{0}, +\infty)\rightarrow[-2\lambda^{2}_{0}, +\infty)$, which is defined as
\begin{align}
\tau=-i\varphi(\lambda)=2\lambda^{2}-4\lambda\lambda_{0}.
\end{align}
Hence we define function $G(\zeta, \tau)$ as
\begin{align}\label{5.18}
G(\zeta, \tau)=\left\{\begin{aligned}
&\frac{(\lambda-i)^{3}}{\lambda-\lambda_{0}}g(\zeta, \tau),\quad \tau\geq-2\lambda^{2}_{0},\\
&0,\quad \tau<-2\lambda^{2}_{0},
\end{aligned}\right.
\end{align}
then the function $G(\zeta, \tau)$ is $C^{11}$ and
\begin{align}
\begin{aligned}
\frac{\partial^{n}G}{\partial\tau^{n}}(\zeta, \tau)
&=\frac{\partial^{n}G}{\partial\lambda^{n}}\left(\frac{\partial\lambda}{\partial\tau}\right)^{n}\\
&=\left(\frac{1}{4(\lambda-\lambda_{0})}\frac{\partial}{\partial\lambda}\right)^{n}
\left[\frac{(\lambda-i)^{3}}{\lambda-\lambda_{0}}g(\zeta, \tau)\right],\quad \tau\geq-2\lambda^{2}_{0}.
\end{aligned}
\end{align}
For this reason, we have
\begin{align}
\left\|\frac{\partial^{n}G}{\partial\tau^{n}}\right\|_{L^{2}(\mathbb{R})}<\infty, \quad n=0,1,2.
\end{align}
In particular, $G(\zeta,\cdot)$ belongs to $H^{2}(\mathbb{R})$.

From the above, Fourier transform $\hat{G}(\zeta, s)$ can be deduced as
\begin{align}
\hat{G}(\zeta, s)=\frac{1}{2\pi}\int_{\mathbb{R}}G(\zeta, \tau)e^{-i\tau s}d\tau,
\end{align}
where
\begin{align}\label{5.23}
G(\zeta, \tau)=\int_{\mathbb{R}}\hat{G}(\zeta, s)e^{i\tau s}ds,
\end{align}
which satisfies the Plancherel theorem, i.e.
\begin{align}
\left\|s^{2}\hat{G}(s)\right\|_{L^{2}(\mathbb{R})}<\infty.
\end{align}

Combining \eqref{5.18} and \eqref{5.23} can get
\begin{align}
\frac{\lambda-\lambda_{0}}{(\lambda-i)^{3}}\int_{\mathbb{R}}\hat{G}(\zeta, s)e^{i\tau s}ds=
\left\{\begin{aligned}
&g(\zeta, \lambda),\quad \lambda\geq\lambda_{0},\\
&0,\quad \lambda<\lambda_{0}.
\end{aligned}\right.
\end{align}
Rewriting $g(\zeta, \lambda)$ as
\begin{align}
g(\zeta, \lambda)=g_{a}(\tilde{x},t;\lambda)+g_{r}(\tilde{x},t;\lambda),
\end{align}
where
\begin{align}
&g_{a}(\tilde{x},t;\lambda)=\frac{\lambda-\lambda_{0}}{(\lambda-i)^{3}}\int^{+\infty}_{-\frac{t}{4}}\hat{G}(\zeta, s)e^{i\tau s}ds,\quad t>0, \lambda\in\bar{D}_{1},\\
&g_{r}(\tilde{x},t;\lambda)=\frac{\lambda-\lambda_{0}}{(\lambda-i)^{3}}\int^{-\frac{t}{4}}_{-\infty}\hat{G}(\zeta, s)e^{i\tau s}ds,\quad t>0, \lambda\geq\lambda_{0},
\end{align}
therefore one implies that $g_{a}(\tilde{x},t;\cdot)$ is analytic in $D_{1}$ and continuous in $\bar{D}_{1}$. Moreover, estimating $g_{a}(\tilde{x},t;\lambda)$ and $g_{r}(\tilde{x},t;\lambda)$ can get
\begin{align}
\begin{aligned}
|g_{a}(\tilde{x},t;\lambda)|&\leq\frac{|\lambda-\lambda_{0}|}{|\lambda-i|^{3}}\|\hat{G}(s)\|_{L^{1}(\mathbb{R})}
\sup_{s\geq-\frac{t}{4}}e^{sRei\tau(\zeta,\lambda)}\\
&\leq\frac{C|\lambda-\lambda_{0}|}{|\lambda-i|^{3}}e^{\frac{t}{4}|Rei\tau(\zeta,\lambda)|},\quad t>0, \lambda\in\bar{D}_{1},
\end{aligned}
\end{align}
and
\begin{align}
\begin{aligned}
|g_{r}(\tilde{x},t;\lambda)|&\leq\frac{|\lambda-\lambda_{0}|}{|\lambda-i|^{3}}
\int^{-\frac{t}{4}}_{-\infty}s^{2}|\hat{G}(\zeta, s)|s^{-2}ds\\
&\leq\frac{C}{1+|\lambda|^{2}}\|s^{2}|\hat{G}(\zeta, s)\|_{L^{2}(\mathbb{R})}
\sqrt{\int^{-\frac{t}{4}}_{-\infty}s^{-4}ds}\\
&\leq\frac{C}{1+|\lambda|^{2}}t^{-\frac{3}{2}},\quad t>0, \lambda\geq\lambda_{0}.
\end{aligned}
\end{align}
Therefore the $L^{1}$, $L^{2}$ and $L^{\infty}$ norm of $g_{r}$ are $O(t^{-\frac{3}{2}})$ on $(\lambda_{0}, \infty)$.

Taking
\begin{align}
&r_{1,a}(\tilde{x},t;\lambda)=g_{0}(\tilde{x},t;\lambda)+g_{a}(\tilde{x},t;\lambda),\quad \lambda\in\bar{D}_{1},\\
&r_{1,r}(\tilde{x},t;\lambda)=g_{r}(\tilde{x},t;\lambda),\quad \lambda\geq\lambda_{0},
\end{align}
then we obtain the decomposition of $r_{1}$. The case of $r_{j}, j=2,3,4,$ can also be studied in the same way.

According to the above analysis, the Lemma is given as follows.
\begin{lem}
For the decomposition of $r(\lambda)$, the related properties are listed below:

 \begin{itemize}
 \item[(i)]
 $r_{j,a}(\tilde{x},t;\lambda)$ are analytic in $D_{j}$ and continuous in $\bar{D}_{j}$ for $t>0$, which satisfy
\begin{align}
\begin{aligned}
\left|r_{j,a}(\tilde{x},t;\lambda)-r_{j}(\lambda)\right|\leq&C|\lambda-\lambda_{0}|
e^{\frac{t}{4}|Rei\tau(\zeta,\lambda)|},\\ &\lambda\in\bar{D}_{j},~ |\lambda|<K,~ t>0,~ j=1,\ldots,4,
\end{aligned}
\end{align}
where each $K>0$ and $C$ is independent of $\zeta, t, \lambda.$

\item[(ii)] $r_{1,a}(\tilde{x},t;\lambda)$ and $r_{2,a}(\tilde{x},t;\lambda)$ satisfy
\begin{align}
|r_{j,a}(\tilde{x},t;\lambda)|\leq\frac{C}{1+|\lambda|}e^{\frac{t}{4}|Rei\tau(\zeta,\lambda)|},
\quad \lambda\in\bar{D}_{j},~t>0,~ j=1,2,
\end{align}
where $C$ is independent of $\zeta, t, \lambda.$

\item[(iii)] The $L^{1}$, $L^{2}$ and $L^{\infty}$ norm of $r_{1,r}$ and $r_{2,r}$ are $O(t^{-\frac{3}{2}})$ on $(\lambda_{0}, +\infty)$ as $t\rightarrow\infty$.

\item[(iv)] The $L^{1}$, $L^{2}$ and $L^{\infty}$ norm of $r_{3,r}$ and $r_{4,r}$ are $O(t^{-\frac{3}{2}})$ on $(-\infty, \lambda_{0})$ as $t\rightarrow\infty$.
\end{itemize}
\end{lem}

Besides, before the deformation of the RH problem, we introduce oriented counter $\Gamma$ and open sets $\{V_j\}_1^6$ as depicted in Fig.2.

\centerline{\begin{tikzpicture}[scale=0.5]
\draw[->][thick](-5,0)--(-2.5,0);
\draw[->][thick](-2.5,0)--(2.5,0);
\draw[-][thick](2.5,0)--(5,0);
\draw[->][thick](-4,4)--(-2,2);
\draw[-][thick](-2,2)--(2,-2);
\draw[<-][thick](2,-2)--(4,-4);
\draw[-][thick](-4,-4)--(-2,-2);
\draw[<->][thick](-2,-2)--(2,2);
\draw[-][thick](-2,-2)--(4,4);
\draw[fill] (2.5,1.25)node[below]{\small{$V_{1}$}};
\draw[fill] (0,2.5)node[below]{\small{$V_{2}$}};
\draw[fill] (-2.5,1.25)node[below]{\small{$V_{3}$}};
\draw[fill] (-2.5,-1.25)node[above]{\small{$V_{4}$}};
\draw[fill] (0,-2.5)node[above]{\small{$V_{5}$}};
\draw[fill] (2.5,-1.25)node[above]{\small{$V_{6}$}};
\end{tikzpicture}}
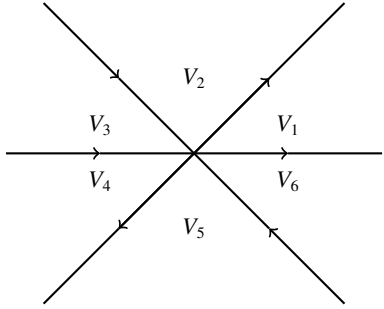
\centerline{\noindent {\small \textbf{Fig. 2} The jump contour $\Gamma$ and the domains $\{V_{j}\}^{6}_{1}$. }}

As a result, defining a transformation
\begin{align}\label{Trans-2}
M^{(2)}(\tilde{x},t;\lambda)=M^{(1)}(\tilde{x},t;\lambda)R,
\end{align}
here
\begin{align}
R=\left\{\begin{aligned}
&\left(
  \begin{array}{cc}
    1 & -e^{-2t\varphi(\lambda)}r_{1,a}(\lambda)\delta^{-2} \\
    0 & 1 \\
  \end{array}
\right), ~&\lambda\in V_{1},\\
&\left(
  \begin{array}{cc}
    1 &  0\\
    e^{2t\varphi(\lambda)}r_{2,a}(\lambda)\delta^{2} & 1 \\
  \end{array}
\right), ~&\lambda\in V_{6},\\
&\left(
  \begin{array}{cc}
    1 & e^{-2t\varphi(\lambda)}r_{3,a}(\lambda)\delta^{-2}_{-} \\
    0 & 1 \\
  \end{array}
\right), ~&\lambda\in V_{4},\\
&\left(
  \begin{array}{cc}
    1 & 0 \\
    -e^{2t\varphi(\lambda)}r_{4,a}(\lambda)\delta^{2}_{+} & 1 \\
  \end{array}
\right),~ &\lambda\in V_{3},\\
&\left(
  \begin{array}{cc}
    1 & 0 \\
    0 & 1 \\
  \end{array}
\right),~ &\lambda\in V_{2}\cup V_{5}.
\end{aligned}
\right.
\end{align}

Then the function $M^{(2)}(\tilde{x},t;\lambda)$ satisfies the RH problem.\\
$\bullet$ $M^{(2)}(\tilde{x},t;\lambda)$ is analytic in $\mathbb{C}\backslash\Gamma$.\\
$\bullet$ The jump condition is
\begin{align}
M^{(2)}_{+}(\tilde{x},t;\lambda)=M^{(2)}_{-}(\tilde{x},t;\lambda)J^{(2)}(\tilde{x},t;\lambda),
\end{align}
with
\begin{align*}
J^{(2)}(\tilde{x},t;\lambda)=\left\{\begin{aligned}
&\left(
  \begin{array}{cc}
    1 & e^{-2t\varphi(\lambda)}r_{1,a}(\lambda)\delta^{-2} \\
    0 & 1 \\
  \end{array}
\right), ~&\lambda\in V_{1}\cap V_{2},\\
&\left(
  \begin{array}{cc}
    1 &  0\\
    -e^{2t\varphi(\lambda)}r_{2,a}(\lambda)\delta^{2} & 1 \\
  \end{array}
\right), ~&\lambda\in V_{5}\cap V_{6},\\
&\left(
  \begin{array}{cc}
    1 & -e^{-2t\varphi(\lambda)}r_{3,a}(\lambda)\delta^{-2}_{-} \\
    0 & 1 \\
  \end{array}
\right), ~&\lambda\in V_{2}\cap V_{3},\\
&\left(
  \begin{array}{cc}
    1 & 0 \\
    e^{2t\varphi(\lambda)}r_{4,a}(\lambda)\delta^{2}_{+} & 1 \\
  \end{array}
\right),~ &\lambda\in V_{4}\cap V_{5},\\
&\left(
  \begin{array}{cc}
    1 & e^{-2t\varphi(\lambda)}r_{1,r}(\lambda)\delta^{-2} \\
    e^{2t\varphi(\lambda)}r_{2,r}(\lambda)\delta^{2} & 1+r_{1,r}(\lambda)r_{2,r}(\lambda) \\
  \end{array}
\right), ~&\lambda\in V_{1}\cap V_{6},\\
&\left(
  \begin{array}{cc}
    1+r_{3,r}(\lambda)r_{4,r}(\lambda) & e^{-2t\varphi(\lambda)}r_{3,r}(\lambda)\delta^{-2}_{-} \\
    e^{2t\varphi(\lambda)}r_{4,r}(\lambda)\delta^{2}_{+} & 1 \\
  \end{array}
\right), ~&\lambda\in V_{3}\cap V_{4}.
               \end{aligned}\right.
\end{align*}
$\bullet$ $M^{(2)}(\tilde{x},t;\lambda)\rightarrow I$, \quad as $\lambda\rightarrow\infty$.

{\bf Step 3}
For the RH problem related to $M^{(2)}(\tilde{x},t;\lambda)$, we introduce a parabolic cylindrical function to transform it into a solvable RH problem.

To begin with, let $X$ denotes $X=X_{1}\cup\ldots\cup X_{4}\subset\mathbb{C}$, where
\begin{align*}
&X_{1}=\{se^{\frac{\pi}{4}i}|0\leq s<\infty\}, \quad X_{2}=\{se^{-\frac{\pi}{4}i}|0\leq s<\infty\},\\
&X_{3}=\{se^{-\frac{3\pi}{4}i}|0\leq s<\infty\}, \quad X_{4}=\{se^{\frac{3\pi}{4}i}|0\leq s<\infty\},
\end{align*}
which are shown as Fig. 3.

\centerline{\begin{tikzpicture}[scale=0.5]
\draw[->][thick](-4,4)--(-2,2);
\draw[-][thick](-2,2)--(2,-2);
\draw[<-][thick](2,-2)--(4,-4);
\draw[-][thick](-4,-4)--(-2,-2);
\draw[<->][thick](-2,-2)--(2,2);
\draw[-][thick](-2,-2)--(4,4);
\draw[fill] (2,2)node[above]{$X_{1}$};
\draw[fill] (-2,2)node[above]{$X_{4}$};
\draw[fill] (2,-2)node[below]{$X_{2}$};
\draw[fill] (-2,-2)node[below]{$X_{3}$};
\end{tikzpicture}}
\centerline{\noindent {\small \textbf{Fig. 3} The jump contour $X$.}}

Next, we introduce a parabolic cylindrical function, i.e. in complex $z$-plane, $m^{X}(\zeta,z)$ is the solution of the RH problem as
\begin{align}
&m^{X}_{+}(\zeta,z)=m^{X}_{-}(\zeta,z)V^{X}(\zeta,z), \quad z\in X,\\
&m^{X}(\zeta,z)\rightarrow I, \quad z\rightarrow\infty,
\end{align}
where the jump matrix is
\begin{align}
V^{X}(\zeta,z)=\left\{\begin{aligned}
&\left(
  \begin{array}{cc}
    1 & q(\zeta)z^{-2i\upsilon(\zeta)}e^{\frac{iz^{2}}{2}} \\
    0 & 1 \\
  \end{array}
\right), \qquad & z \in X_{1},\\
&\left(
  \begin{array}{cc}
    1 &  0\\
    q^{*}(\zeta)z^{2i\upsilon(\zeta)}e^{-\frac{iz^2}{2}} & 1 \\
  \end{array}
\right), \qquad & z \in X_{2},\\
&\left(
  \begin{array}{cc}
    1 & -\frac{q(\zeta)}{1+|q(\zeta)|^{2}}z^{-2i\nu(\zeta)}e^{\frac{iz^2}{2}} \\
    0 & 1 \\
  \end{array}
\right), \qquad & z \in X_{3},\\
&\left(
  \begin{array}{cc}
    1 & 0 \\
    -\frac{q^{*}(\zeta)}{1+|q(\zeta)|^{2}}z^{2i\nu(\zeta)}e^{-\frac{iz^2}{2}} & 1 \\
  \end{array}
\right), \qquad & z \in X_{4}.
\end{aligned}\right.
\end{align}

Moreover, matching $V^{X}(\zeta,z)$ to $J^{(2)}(\tilde{x},t;\lambda)$ and calculating directly can get that $q(\zeta)$ is
\begin{align}
q(\zeta)=e^{-2\chi(\zeta,\lambda_{0})}r(\lambda_{0})e^{2i\upsilon(\zeta)\ln (\sqrt{-8})}.
\end{align}

For the RH problem related to $m^{X}(\zeta,z)$, the following Lemma are directly given.
\begin{lem}
The unique solution for the RH problem related to $m^{X}(\zeta,z)$  is
\begin{align}
m^{X}(\zeta,z)=I+\frac{i}{z}\left(\begin{array}{cc}
                                    0 & \gamma^{X}(\zeta) \\
                                    -\beta^{X}(\zeta) & 0
                                  \end{array}\right)
+O\left(\frac{1}{z^{2}}\right), \quad z\rightarrow\infty,
\end{align}
where
\begin{align}
&\beta^{X}(\zeta)=\frac{\sqrt{2\pi}e^{\frac{\pi}{4}i}e^{-\frac{\pi}{2}\upsilon(\zeta)}}
{q(\zeta)\Gamma(-i\upsilon(\zeta))},\\
&\gamma^{X}(\zeta)=\frac{\sqrt{2\pi}e^{-\frac{\pi}{4}i}e^{-\frac{\pi}{2}\upsilon(\zeta)}}
{q^{*}(\zeta)\Gamma(i\upsilon(\zeta))}.
\end{align}
\end{lem}

\begin{proof}
For details, please refer to Appendix B in \cite{22}.
\end{proof}

According to the above Lemma, we know that the Riemann-Hilbert problem related to $m^{X}(\zeta,z)$ can be solved accurately with a parabolic cylinder function. Next, we further set
\begin{align}
D(\zeta,t)=e^{-(2it\lambda^{2}_{0}-\chi(\lambda_{0}))\sigma_{3}}(-8t)^{-\frac{i\upsilon(\lambda_{0})}{2}\sigma_{3}},
\end{align}
and define $m_{0}(\tilde{x},t;\lambda)$ for $\lambda$ near $\lambda_{0}$ as
\begin{align}
m_{0}(\tilde{x},t;\lambda)=D(\zeta,t)m^{X}\left(\zeta,\sqrt{-8t}(\lambda-\lambda_{0})\right)D^{-1}(\zeta,t).
\end{align}

Using $m_{0}(\tilde{x},t;\lambda)$ to introduce function $M^{(3)}(\tilde{x},t;\lambda)$
\begin{align}
M^{(3)}(\tilde{x},t;\lambda)=\left\{\begin{aligned}
&M^{(2)}(\tilde{x},t;\lambda)m^{-1}_{0}(\tilde{x},t;\lambda), \quad |\lambda-\lambda_{0}|<\epsilon,\\
&M^{(2)}(\tilde{x},t;\lambda), \quad elsewhere,
\end{aligned}\right.
\end{align}
where $\epsilon=\lambda_{0}/2$.

\centerline{\begin{tikzpicture}[scale=0.5]
\draw[->][thick](-5,0)--(-2.5,0);
\draw[->][thick](-2.5,0)--(2.5,0);
\draw[-][thick](2.5,0)--(5,0);
\draw[->][thick](-4,4)--(-2,2);
\draw[-][thick](-2,2)--(2,-2);
\draw[<-][thick](2,-2)--(4,-4);
\draw[-][thick](-4,-4)--(-2,-2);
\draw[<->][thick](-2,-2)--(2,2);
\draw[-][thick](-2,-2)--(4,4);
\draw[->][thick] (0,1.5) arc (90:270:1.5);
\draw[->][thick] (0,-1.5) arc (-90:90:1.5);
\draw[fill] (0,0)node[below]{$\lambda_{0}$};
\end{tikzpicture}}
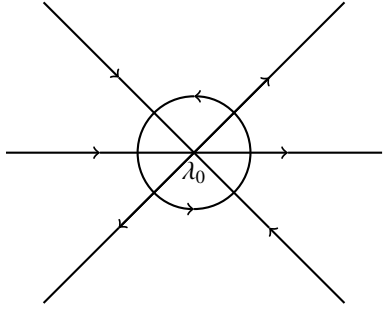
\centerline{\noindent {\small \textbf{Fig. 4}   The jump contour $\Gamma^{(3)}$.}}

According to the definition of $M^{(2)}(\tilde{x},t;\lambda)$ and $m_{0}(\tilde{x},t;\lambda)$, we know that $M^{(3)}(\tilde{x},t;\lambda)$ admits the Riemann-Hilbert problem
\begin{align}
\left\{\begin{aligned}
&M^{(3)}_{+}(\tilde{x},t;\lambda)=M^{(3)}_{-}(\tilde{x},t;\lambda)J^{(3)}(\tilde{x},t;\lambda), \quad \lambda\in\Gamma^{(3)},\\
&M^{(3)}_{+}(\tilde{x},t;\lambda)\rightarrow I, \quad \lambda\rightarrow\infty,
\end{aligned}\right.
\end{align}
where the jump matrix is
\begin{align}
J^{(3)}(\tilde{x},t;\lambda)=
\left\{\begin{aligned}
&m_{0,-}(\tilde{x},t;\lambda)J^{(2)}(\tilde{x},t;\lambda)m^{-1}_{0,+}(\tilde{x},t;\lambda), \quad |\lambda-\lambda_{0}|<\epsilon,\\
&m^{-1}_{0}(\tilde{x},t;\lambda), \quad |\lambda-\lambda_{0}|=\epsilon,\\
&J^{(2)}(\tilde{x},t;\lambda), \quad elsewhere.
\end{aligned}\right.
\end{align}

\section{Long time asymptotics}

In this section, we will discuss the long-time asymptotic solution of the WKI equation using the model Riemann-Hilbert problem.

Let $\omega:=J^{(3)}-I$. Combining Theorem 2.1 of \cite{22} can get the estimates for the columns of $\omega(x,t;\lambda)$ which allow
\begin{align}
\omega^{(j)}(\zeta,t,\cdot)=O(t^{-\frac{1}{2}\ln t}), \quad t\rightarrow\infty,
\lambda\in \Gamma\cap\{|\lambda-\lambda_{0}|<\epsilon\}, j=1,2.
\end{align}
Moreover,
\begin{align}
&\|\omega^{(j)}(\zeta,t,\cdot)\|_{L^{2}(\Gamma^{(3)})}=O(t^{-\frac{1}{2}}), \quad t\rightarrow\infty,j=1,2\\
&\|\omega(\zeta,t,\cdot)\|_{L^{1}\cap L^{2}(\Gamma^{(3)})}=O(t^{-\frac{1}{2}}), \quad t\rightarrow\infty,\label{5.3}\\
&\|\omega(\zeta,t,\cdot)\|_{L^{\infty}(\Gamma^{(3)})}=O(t^{-\frac{1}{2}}\ln t), \quad t\rightarrow\infty.\label{5.4}
\end{align}

According to $C_{\omega}f=C_{-}(f\omega)$, where $C_{-}(f\omega)$ is the boundary value of $C(f\omega)$ from the right side of $\Gamma^{(3)}$ and $C$ is the Cauchy operator related to $\Gamma^{(3)}$ as
\begin{align*}
(Cf)(z)=\frac{1}{2\pi i}\int_{\Gamma^{(3)}}\frac{f(s)}{s-z}ds, \quad z\in\mathbb{C}\setminus\Gamma^{(3)},
\end{align*}
the integral operator is defined as
\begin{align}
C_{\omega}:L^{2}(\Gamma^{(3)})+L^{\infty}(\Gamma^{(3)})\rightarrow L^{2}(\Gamma^{(3)}).
\end{align}

Using \eqref{5.3} and \eqref{5.4}, we have
\begin{align*}
\|C_{\omega}\|_{B(L^{2}(\Gamma^{(3)}))}\leq C\|\omega\|_{L^{\infty}(\Gamma^{(3)})}=O(t^{-1}\ln t),
\end{align*}
where $B(L^{2}(\Gamma^{(3)}))$ means the Banach space of bounded linear operators $L^{2}(\Gamma^{(3)})\rightarrow L^{2}(\Gamma^{(3)})$. Since $\|C_{\omega}\|_{B(L^{2}(\Gamma^{(3)}))}\rightarrow0$ as $t\rightarrow\infty$, there exists $T>0$ such that $1-C_{\omega}\in B(L^{2}(\Gamma^{(3)}))$ is invertible.

Due to \eqref{5.3}, one has
\begin{align}
\|\mu(\zeta,t,\cdot)-I\|_{L^{2}(\Gamma^{(3)})}=O(t^{-\frac{1}{2}}), \quad t\rightarrow\infty,
\end{align}
where $\mu-I$ can be regarded as the solution of the equation
\begin{align}
(I-C_{\omega})(\mu-I)=C_{\omega}I,
\end{align}
then
\begin{align}
M^{(3)}(\tilde{x},t;\lambda)=I+C(\mu\omega)=
I+\frac{1}{2\pi i}\int_{\Gamma^{(3)}}\mu(\zeta,t;s)\omega(\zeta,t;s)\frac{ds}{s-\lambda},
\end{align}
is the unique solution of the RH problem.

When $\lambda\rightarrow0$, $M^{(3)}(\tilde{x},t;\lambda)$ is expanded as
\begin{align}
M^{(3)}(\tilde{x},t;\lambda)=M^{(3)}_{0}+M^{(3)}_{1}\lambda+O(\lambda^{2}),
\end{align}
where
\begin{align}
&M^{(3)}_{0}=M^{(3)}(\tilde{x},t;0)=I+\frac{1}{2\pi i}\int_{\Gamma^{(3)}}\mu\omega\frac{ds}{s},\\
&M^{(3)}_{1}=\left(M^{(3)}(\tilde{x},t;\lambda)\right)'=\frac{1}{2\pi i}\int_{\Gamma^{(3)}}\mu\omega\frac{ds}{s^{2}}.
\end{align}

For the solution of the WKI equation, one has
\begin{align}
\begin{aligned}
e^{-2d}q(x,t)&=e^{-2d}\tilde{q}(\tilde{x},t)\\
&=\lim_{\lambda\rightarrow0}\frac{\partial}{\partial\tilde{x}}
\frac{\left(\left(\tilde{M}(\tilde{x},t;0)\right)^{-1}\tilde{M}(\tilde{x},t;\lambda)\right)_{12}}{i\lambda}\\
&=-i\lim_{\lambda\rightarrow0}\frac{\partial}{\partial\tilde{x}}
\left(\frac{\left(M^{(3)}(\tilde{x},t;0)\right)^{-1}M^{(3)}(\tilde{x},t;\lambda)-I}{\lambda}\right)_{12}\\
&=-i\lim_{\lambda\rightarrow0}\frac{\partial}{\partial\tilde{x}}
\left(\left(M^{(3)}_{0}\right)^{-1}M^{(3)}_{1}\right)_{12}.
\end{aligned}
\end{align}

On the one hand,
\begin{align}
\begin{aligned}
M^{(3)}_{1}&=\frac{1}{2\pi i}\int_{\Gamma^{(3)}}\mu(s)\omega(s)\frac{ds}{s^{2}}\\
&=\frac{1}{2\pi i}\left(\int_{|\lambda-\lambda_{0}|=\epsilon}+\int_{\Gamma}\right)\frac{\mu\omega}{s^{2}}ds\\
&=\frac{1}{2\pi i}\int_{|\lambda-\lambda_{0}|=\epsilon}\frac{\mu(m^{-1}_{0}-I)}{s^{2}}ds-
\frac{1}{2\pi i}\int_{\Gamma}\frac{\mu\omega}{s^{2}}ds.
\end{aligned}
\end{align}

For above formula, we have
\begin{align}
\begin{aligned}
m^{-1}_{0}&=D(\zeta,t)m^{X}\left(\zeta,\sqrt{-8t}(\lambda-\lambda_{0})\right)D^{-1}(\zeta,t)\\
&=I+\frac{B(\zeta,t)}{\sqrt{-8t}(\lambda-\lambda_{0})}+O(t^{-1}),
\end{aligned}
\end{align}
where
\begin{align}
B(\zeta,t)=\left(\begin{array}{cc}
0 & \gamma^{X}e^{-2(2it\lambda^{2}_{0}-\chi(\lambda_{0}))}(-8t)^{-i\upsilon(\lambda_{0})} \\
-i\beta^{X}e^{2(2it\lambda^{2}_{0}-\chi(\lambda_{0}))}(-8t)^{i\upsilon(\lambda_{0})} & 0
                 \end{array}\right),
\end{align}
then
\begin{align}
\begin{aligned}
&\int_{|\lambda-\lambda_{0}|=\epsilon}\frac{\mu(m^{-1}_{0}-I)}{s^{2}}ds\\
&=\int_{|\lambda-\lambda_{0}|=\epsilon}(m^{-1}_{0}-I)d\lambda+
\int_{|\lambda-\lambda_{0}|=\epsilon}(\mu-I)(m^{-1}_{0}-I)d\lambda\\
&=\frac{2\pi iB(\zeta,t)}{\sqrt{-8t}}+O(t^{-1}).
\end{aligned}
\end{align}

And
\begin{align}
\begin{aligned}
\left|\int_{\Gamma}\frac{\mu\omega}{s^{2}}ds\right|
&=\left|\int_{\Gamma}(\mu-I)\frac{\omega}{s^{2}}ds+\int_{\Gamma}\frac{\omega}{s^{2}}ds\right|\\
&\leq\|\mu-I\|_{L^{2}(\Gamma)}\|\frac{\omega}{s^{2}}\|_{L^{2}(\Gamma)}+\|\frac{\omega}{s^{2}}\|_{L^{1}(\Gamma)}\\
&=O(t^{-1}\ln t),
\end{aligned}
\end{align}
thus
\begin{align}
M^{(3)}_{1}=\frac{B(\zeta,t)}{\sqrt{-8t}}+O(t^{-1}\ln t).
\end{align}

On the other hand,
\begin{align}
M^{(3)}_{0}=I+\frac{1}{2\pi i}\int_{\Gamma^{(3)}}\mu\omega\frac{ds}{s},
\end{align}
then
\begin{align}
M^{(3)}_{0}-I=\frac{1}{2\pi i}\int_{\Gamma^{(3)}}\mu\omega\frac{ds}{s}.
\end{align}

In the same way, we have
\begin{align}
M^{(3)}_{0}=I+\frac{B(\zeta,t)}{\sqrt{-8t}}+O(t^{-1}\ln t).
\end{align}

In summary,
\begin{align}
\begin{aligned}
e^{-2d}q(x,t)&=e^{-2d}\tilde{q}(\tilde{x},t)\\
&=\frac{-i\frac{\partial}{\partial\tilde{x}}\gamma^{X}e^{-2(2it\lambda^{2}_{0}-\chi(\lambda_{0}))}
(-8t)^{-(\frac{1}{2}+i\upsilon(\lambda_{0}))}}{1-\frac{i\gamma^{X}\beta^{X}}{8t}}+O(t^{-1}\ln t).
\end{aligned}
\end{align}

Therefore, the following theorem is summarized.

\begin{thm}
Considering the initial value condition \eqref{1.2}, the solution $q(x,t)$ of the WKI equation is solved as
\begin{align}
q(x,t)=e^{2d}\frac{-i\frac{\partial}{\partial\tilde{x}}\gamma^{X}e^{-2(2it\lambda^{2}_{0}-\chi(\lambda_{0}))}
(-8t)^{-(\frac{1}{2}+i\upsilon(\lambda_{0}))}}{1-\frac{i\gamma^{X}\beta^{X}}{8t}}+O(t^{-1}\ln t),\quad t\rightarrow\infty,
\end{align}
where the error term is uniform with respect to $x$.
\end{thm}

\section{Conclusions and discussions}

In this work, under zero boundary condition, the Deift-Zhou nonlinear steepest descent method for the oscillatory RH problem is used to study the Wadati-Konno-Ichikawa equation. What is different from the past is that in the spectral analysis, we not only discuss the case of $\lambda\rightarrow\infty$, but also the case of $\lambda\rightarrow0$. Then their properties are used to establish a connection, and the solution $q(x,t)$ of the WKI equation is expressed by the solution of the RH problem. After performing a series of deformations on the oscillatory RH problem, we further get a model RH problem. Finally, the long time asymptotic solution $q(x,t)$ of the WKI equation is derived, which greatly enriches the research on the solution of the WKI equation, and gives us a deeper understanding for the long time asymptotic solution of non-AKNS system.

\section*{Acknowledgements}

This work was supported by  the National Natural Science Foundation of China under Grant No. 11975306, the Natural Science Foundation of Jiangsu Province under Grant No. BK20181351, the Six Talent Peaks Project in Jiangsu Province under Grant No. JY-059, and the Fundamental Research Fund for the Central Universities under the Grant Nos. 2019ZDPY07 and 2019QNA35.



\begin{thebibliography}{99}

\bibitem{1}
V.E. Zakharov, A.B. Shabat, Exact theory of two-dimensional self-focusing and one-dimensional self-modulation of waves in nonlinear media, Sov. Phys. JETP 34 (1972) 62-69.
\bibitem{Zakharov-1968}
V.E. Zakharov, Stability of periodicwaves of finite amplitude on the surface of a deep fluid, Sov. Phys. J. Appl. Mech. Tech. Phys. 4 (1968) 190-194.
\bibitem{Zakharov-1972}
V.E. Zakharov, Collapse of Langmuir waves, Sov. Phys. JETP 35 (1972) 908-914.
\bibitem{12}
D.J. Kaup, A.C. Newell, An exact solution for a derivative nonlinear Schr\"{o}dinger equation, J. Math. Phys. 19 (1978) 798-801.
\bibitem{13}
A.K. Zvezdin, A.F. Popkov, Contribution to the nonlinear theory of magnetostat-icspin waves, Sov. Phys. JETP 2 (1983) 350.

\bibitem{11}
S.F. Tian, Initial-boundary value problems for the general coupled nonlinear Schr\"{o}dinger equation on the interval via the Fokas method, J. Differ. Equ. 262 (2017) 506-558.
\bibitem{2}
W.X. Ma, Inverse scattering and soliton solutions of nonlocal reverse-spacetime nonlinear Schr\"{o}dinger equations, Proc. Amer. Math. Soc. 149(1) (2021) 251-263.
\bibitem{3}
W.Q. Peng, S.F. Tian, X.B. Wang, T.T. Zhang, Y. Fang, Riemann-Hilbert method and multi-soliton solutions for three-component coupled nonlinear Schr\"{o}dinger equations, J. Geom. Phys. 146 (2019) 103508.
\bibitem{10}
X.G. Geng, H. Liu, J.Y. Zhu, Initial-boundary value problems for the coupled nonlinear Schr\"{o}dinger equation on the half-line, Stud. Appl. Math. (2015) 310-346.
\bibitem{5}
A. Chowdury, D.J. Kedziora, A. Ankiewicz, N. Akhmediev, Soliton solutions of an integrable nonlinear Schr\"{o}dinger equation with quintic terms, Phys. Rev. E 90 (2014) 032922.
\bibitem{Fan-2020}
P. Zhao, E.G. Fan, Finite gap integration of the derivative nonlinear Schr\"{o}dinger equation: A Riemann-Hilbert method, Phys. D 402(15) (2020) 132213.
\bibitem{4}
Z.Q. Li, S.F. Tian, W.Q. Peng, J.J. Yang, Inverse scattering transform and soliton classifification of the higher order nonlinear Schr\"{o}dinger-Maxwell-Bloch equations, Theor. Math. Phys. 203(3) (2020) 709-725.




\bibitem{Herrmann-1991}
J. Herrmann, Propagation of ultrashort light pulses in fibers with saturable nonlinearity in the normal-dispersion region, J. Opt. Soc. Amer. B 8 (1991) 1507-1511.
\bibitem{Porsezian-2012}
K. Porsezian, K. Nithyanandan, R.V.J. Raja, P.K. Shukla, Modulational instability at the proximity of zero dispersion wavelength in the relaxing saturable nonlinear system, J. Opt. Soc. Amer. B 29 (2012) 2803-2813.

\bibitem{Wadati-1979}
M. Wadati, K. Konno, and Y.-H. Ichikawa, A Generalization of Inverse Scattering Method, J. Phys. Soc. Jpn. 46 (1979) 1965.
\bibitem{Shimizu-1980}
T. Shimizu and M. Wadati, A New Integrable Nonlinear Evolution Equation, Prog. Theor. Phys. 63 (1980) 808.

\bibitem{Ishimori-1982}
Y. Ishimori, A relationship between the Ablowitz-Kaup-Newell-Segur and Wadati-Konno-Ichikawa schemes of the inverse scattering method, J. Phys. Soc. Jpn. 51 (1982) 3036-3041.

\bibitem{Boiti-1986}
M. Boiti, V.S. Gerdjikov, F. Pempinelli, The WKIS system: B\"{a}cklund transformations, generalized Fourier transforms and all that, Progr. Theoret. Phys. 75 (1986) 1111-1141.
\bibitem{Kundu-1987}
A. Kundu, Explicit auto-B\"{a}cklund relation through gauge transformation, J. Phys. A: Math. Gen. 20 (1987) 1107-1114.
\bibitem{Xiao-1990}
Y. Xiao, The B\"{a}cklund transformations for the Wadati-Konno-Ichikawa system, Phys. Lett. A 149 (1990) 369-370.

\bibitem{Choudhury-1982}
J. Choudhury, Inverse scattering method for a new integrable nonlinear evolution equation under nonvanishing boundary conditions, J. Phys. Soc. Japan 51 (1982) 2312-2317.


\bibitem{Van1-2013}
R.A. Van Gorder, Orbital stability for stationary solutions of the Wadati-Konno-Ichikawa-Shimizu equation, J. Phys. Soc. Japan 82 (2013) 064005.
\bibitem{Van2-2013}
R.A. Van Gorder, Exact stationary solution method for the Wadati-Konno-Ichikawa-Shimizu (WKIS) equation, Progr. Theoret. Phys. 128 (2013) 993-999.
\bibitem{Shimabukuro-2016}
Y. Shimabukuro, Global solution of the Wadati-Konno-Ichikawa equation with small initial data, arXiv:1612.07579v1, 2016.
\bibitem{Li-2016}
Z. Li, X. Geng, L. Guan, Algebro-geometric constructions of the Wadati-Konno-Ichikawa flows and applications, Math. Methods Appl. Sci. 39 (2016) 734-743.

\bibitem{Zhang-2017}
Y. Zhang, D. Qiu, Y. Cheng, J. He, The Darboux transformation for the Wadati-Konno-Ichikawa system, Theor. Math. Phys. 191 (2017) 710-724.

\bibitem{Liu-2017}
H.F. Liu, Y. Shimabukuro, $N$-soliton formula and blowup result of the Wadati-Konno-Ichikawa equation, J. Phys. A 50 (2017) 315204.
\bibitem{Zhang-2019}
Y.S. Zhang, J.G. Rao, Y. Cheng, J.S. He, Riemann-Hilbert method for the Wadati-Konno-Ichikawa equation: $N$ simple poles and one higher-order pole, Phys. D 399 (2019) 173-185.

\bibitem{Zakharov-1976}
V.E. Zakharov and S.V. Manakov, Asymptotic behavior of nonlinear wave systems integrated by the inverse scattering method, Sov. Phys. JETP 44 (1976) 106-112.

\bibitem{Deift-1993}
P. Deift and X. Zhou, A steepest descent method for oscillatory Riemann-Hilbert problems. Asymptotics for the MKdV equation, Ann. Math. 137(2) (1993) 295-368.


\bibitem{deNLS}
P. Deift, A.R. Its, and X. Zhou, Long-time asymptotics for integrable nonlinear wave equations, Important developments in soliton theory, Nonlinear Dynam. (1993) 181-204.

\bibitem{SG}
P.J. Cheng, S. Venakides, X. Zhou, Long-time asymptotics for the pure radiation solution of the sine-Gordon equation, Commun. Part. Diff. Equ. 24 (1999) 1195-1262.

\bibitem{mNLS1}
A.V. Kitaev, A.H. Vartanian, Leading-order temporal asymptotics of the modified nonlinear Schr\"{o}dinger equation: solitonless sector, Inverse Problems 13 (1997) 1311-1339.
\bibitem{mNLS2}
A.V. Kitaev, A.H. Vartanian, Asymptotics of solutions to the modified nonlinear Schr\"{o}dinger equation: solution on a nonvanishing continuous background, SIAM J. Math. Anal. 30 (1999) 787-832.

\bibitem{fNLS}
G. Biondini, D. Mantzavinos, Long-time asymptotics for the focusing nonlinear Schr\"{o}dinger equation with nonzero boundary conditions at infinity and asymptotic stage of modulational instability, Comm. Pure Appl. Math. 70(12) (2017) 2300-2365.
\bibitem{KdV}
K. Grunert, G. Teschl, Long-time asymptotics for the Korteweg-de Vries equation via nonlinear steepest descent, Math. Phys. Anal. Geom. 12 (2009) 287.
\bibitem{dNLS}
J. Xu, E. G. Fan and Y. Chen, Long-time Asymptotic for the Derivative Nonlinear Schr\"{o}dinger
Equation with Step-like Initial Value, Math. Phys. Anal. Geom. 16 (2013) 253-288.
\bibitem{CH}
A. Boutet de Monvel, A. Kostenko, D. Shepelsky, G. Teschl, Long-time asymptotics for the Camassa-Holm equation, SIAM J. Math. Anal. 41 (2009) 1559-1588.
\bibitem{FL}
J. Xu and E. G. Fan, Long-time asymptotics for the Fokas-Lenells equation with decaying initial value problem: Without solitons, J. Differ. Equ. 259 (2015) 1098-1148.

\bibitem{GP}
X. Geng, K. Wang, M. Chen, Long-Time Asymptotics for the Spin-1 Gross-Pitaevskii Equation, Commun. Math. Phys. (2021) 1-27.
\bibitem{cHirota}
N. Liu, B. Guo, Long-time asymptotics for the initial-boundary value problem of coupled Hirota equation on the half-line, Sci. China Math. 64(1) (2020).

\bibitem{SP}
J. Xu, Long-time asymptotics for the short pulse equation, J. Differ. Equ. 265(8) (2018) 3494-3532.
\bibitem{CSP}
J. Xu and E. G. Fan, Long-time asymptotic behavior for the complex short pulse equation, J. Differ. Equ. 269(11) (2020) 10322-10349.
\bibitem{KE1}
B. Guo and N. Liu, Long-time asymptotics for the Kundu-Eckhaus equation on the half-line, J. Math. Phys. 59 (2018) 061505.
\bibitem{KE2}
D.S. Wang, B. Guo, X. Wang, Long-time asymptotics of the focusing Kundu¨CEckhaus equation with nonzero boundary conditions, J. Differ. Equ. (2018) http://doi.org/10.1016/j.jde.2018.10.053.
\bibitem{AB}
S. Chen, Z. Yan, Long-time asymptotics of solutions for the coupled dispersive AB system with initial value problems, J. Math. Anal. Appl. (2020) 124401.
\bibitem{emKdV}
N. Liu, B. Guo, D. S. Wang, Y. F. Wang, Long-time asymptotic behavior for an extended modified Korteweg-de Vries equation, Commun. Math. Sci. 17 (2019) 1877-1913.
\bibitem{gB}
C. Charlier, J. Lenells, D. S. Wang, The ``good" Boussinesq equation: long-time asymptotics. 	arXiv:2003.04789.
\bibitem{nmKdV}
F.J. He , E.G. Fan, J. Xu, Long-Time Asymptotics for the Nonlocal MKdV Equation, Commun. Theor. Phys. 71(5) (2019) 475-488.

\bibitem{22}
J. Lenells, The nonlinear steepest descent method for Riemann-Hilbert problems of low regularity, Indiana Univ. Math. J. 66 (2017) 1287-1332.







\end{thebibliography}
\end{document}